\documentclass[a4paper]{article}
\usepackage{dcolumn}
\usepackage[utf8]{inputenc}
\usepackage[english]{babel}
\usepackage[T1]{fontenc}
\usepackage{amsthm}
\usepackage{amsmath,amssymb}
\usepackage{thmtools}
\usepackage{authblk}
\usepackage[caption=false]{subfig}
\usepackage{mathrsfs}
\usepackage[margin=1in]{geometry}
\usepackage{bbm}

\usepackage{graphicx}

\numberwithin{equation}{section}

\newcommand\di[1]{\:\textrm{d}#1}
\newcommand{\R}{\ensuremath{\mathbb{R}}}
\newcommand{\C}{\ensuremath{\mathbb{C}}}
\newcommand{\diag}{\operatorname{diag}}

\newcommand{\imag}{\operatorname{Im}}

\declaretheoremstyle[
spaceabove=\medskipamount, spacebelow=\medskipamount,
headfont=\bfseries,
notefont=\bfseries\boldmath, notebraces={(}{)},
bodyfont=\itshape,
postheadspace=.5em
]{cursive}

\declaretheorem[style=cursive,name=Theorem,numberwithin=section]{thm}
\declaretheorem[style=cursive,name=Lemma,numberlike=thm]{lem}
\declaretheorem[style=cursive,name=Corollary,numberlike=thm]{corollary}
\declaretheorem[style=cursive,name=Proposition,numberlike=thm]{prop}

\declaretheoremstyle[
spaceabove=\medskipamount, spacebelow=\medskipamount,
headfont=\bfseries,
notefont=\bfseries\boldmath, notebraces={(}{)},
bodyfont=\rmfamily,
postheadspace=.5em
]{upright}

\title{On the existence of impurity bound excitons in one-dimensional systems with zero range interactions} 
\author[1,2]{Jonas Have}
\author[3]{Hynek Kova\v{r}\'{\i}k}
\author[2]{Thomas G. Pedersen}
\author[1]{Horia D. Cornean}

\affil[1]{Department of Mathematical Sciences, Aalborg University}
\affil[2]{Department of Physics and Nanotechnology, Aalborg University}
\affil[3]{DICATAM, Sezione di Matematica,
Universit{\`a}  degli studi di Brescia}
\date{}

\begin{document}
\newlength\fheight
\newlength\fwidth
\maketitle 
\begin{abstract}
We consider a three-body one-dimensional Schr\"odinger operator with zero range potentials, which models a positive impurity with charge $\kappa > 0$ interacting with an exciton. We study the existence of discrete eigenvalues as $\kappa$ is varied. On one hand, we show that for sufficiently small $\kappa$ there exists a unique bound state whose binding energy behaves like $\kappa^4$, and we explicitly compute its leading coefficient. On the other hand, if $\kappa$ is larger than some critical value then the system has no bound states.
\end{abstract}


\section{Introduction}\label{sec:intro}
In this paper we consider a system of three one-dimensional non-relativistic quantum particles with zero range interactions. The system models an impurity interacting with an exciton, which is a pair made of an electron and a hole in either a semiconductor or an insulator. We want to give a rigorous description of the existence of bound states in the cases where the impurity has either a small or a large charge.  In the small charge case we prove the existence of a non-degenerate groundstate, we explicitly compute its leading order behavior and compare it to numerical calculations. In the case of a large impurity charge we prove the existence of a critical charge above which the discrete spectrum is absent, and we compute it numerically. The proofs of our main results are based on a combined application of the Feshbach inversion formula and the Birman-Schwinger principle.

The bound states of a helium like system with two negatively charged particles and a positively charged nucleus interaction through zero range potentials were previously examined in \cite{rosenthal1971solution} and in \cite{cornean2006critical}, while the bound states of a system with a negatively charged particle and two positively charged particles with infinite mass were examined in \cite{hogreve2009delta}. Also, the spectral properties of the similar, but more realistic, three-body Coulomb systems in three dimensions have been examined in \cite{martin1995stability,martin2004stability,gridnev2005stability}.

The choice of Coulomb interaction potential in one-dimensional systems is a non-trivial one. The Schrödinger operator for the one-dimensional hydrogen atom with the $1/|x|$ Coulomb potential is not essentially self-adjoint but has an infinite number of self-adjoint extensions, and the choice of extension and corresponding spectral properties are still the subject of active research\cite{de2012mathematical,de2009self,hogreve2014one}. Other options are to modify the Coulomb potential to get rid of the singularity\cite{fischer1995functional} or use zero range interactions, as used in the present paper. One-dimensional systems and zero range interactions might seem unphysical, but in many cases they can be used as toy models in order to avoid complicated numerical computations. In fact some three-dimensional Coulomb systems and one-dimensional systems with zero range interactions share important spectral properties. A classical example is the analogy between the one-dimensional hydrogen atom and the true three-dimensional hydrogen atom as described in  \cite{frost1956delta}.

Also, such simplified models naturally emerge as effective models for higher-dimensional systems submitted to various forms of confinement, like for example the one-dimensional effective models for excitons in carbon nanotubes in \cite{cornean2004one,ronnow2009stability} , one-dimensional models of optical response in one-dimensional semiconductors in \cite{pedersen2015analytical}, and the effective model for atoms in strong magnetic fields in \cite{brummelhuis2002effective,beau2010h2}. In a similar fashion, the system we consider in this paper can be interpreted as a model for impurity bound excitons in a one-dimensional semiconductor using the Wannier model. Excitonic effects are known to have a significant impact on the optical properties of semiconductors\cite{albrecht1998ab}, especially in one- and two-dimensional semiconductors where the reduced screening leads to large exciton binding energies compared to the bandgap. For a thorough introduction to systems with zero range potentials we refer to the book in \cite{albeverio2012solvable}.

The paper is structured as follows. In Section \ref{sec:main-results} we present the model and comment on the main results of the paper. In Section \ref{sec:framework} we specify the framework and introduce some important notation. In Section \ref{sec:small_kappa} we prove our first main result, namely that there exists a single discrete eigenvalue for sufficiently small impurity charge. In Sections \ref{sec:large_kappa} and \ref{sec:proof_of_corollary_col:kappa_c} we prove our second main result about the disappearance of the discrete spectrum if $\kappa$ becomes supercritical. 

\section{The Model and The Main Results}\label{sec:main-results}
Consider the system of two equal but oppositely charged particles with charge $\pm 1$ and mass $m$, and an impurity with charge  $\kappa$ and mass $M$. Let $\sigma = m/(m+M)$ denote the mass fraction, $0 \leq \sigma < 1$. Using relative atomic coordinates and removing its center of the mass, the system is formally described by the Schr\"odinger operator
\begin{equation}\label{eq:H_kappa}
	H_{\kappa,\sigma} = -\frac{1}{2}\Delta-\sigma\partial_x\partial_y-\delta(x-y)+\kappa\delta(x)-\kappa \delta(y),
\end{equation}
on $L^2(\R^2)$, where $\Delta$ is the two-dimensional Laplace operator, $\delta$ is the Dirac delta distribution.

The discrete spectrum of $H_{\kappa,\sigma}$ corresponds to impurity localized excitons. In the following we state our results regarding the discrete spectrum of $H_{\kappa,\sigma}$ and prove them in Secs. \ref{sec:small_kappa}, \ref{sec:large_kappa}, and \ref{sec:proof_of_corollary_col:kappa_c}. The situation is sketched in Figure \ref{fig:figure_1a} where we see the ground state energy and the essential spectrum for $\sigma=0$. The essential spectrum of $H_{\kappa,0}$ will be derived in Section \ref{sec:framework}, but as illustrated by the shaded area in the figure, its bottom stays equal to $-1/4$ on the closed interval $[0,1/\sqrt{2}]$, while for larger $\kappa$ it equals $-\kappa^2/2$.

The first result concerns the existence and behaviour of a discrete eigenvalue of $H_{\kappa,0}=: H_\kappa$ when $\kappa \in (0,1/\sqrt{2}]$.
\begin{thm}\label{thm:main_1}
If $\kappa>0$ is sufficiently small, the operator $H_\kappa$ has precisely one discrete eigenvalue and its leading order behaviour is:
\begin{equation}
	E(\kappa)= -\frac{1}{4}-16\left(\frac{4}{\pi}-1\right)^2\kappa^4 +\mathcal{O}(\kappa^5).
\end{equation}
Furthermore, the energy $E(\kappa)$ is non-degenerate and decreasing if $\kappa \in (0,1/\sqrt{2}]$, hence the operator $H_\kappa$ has at least one discrete eigenvalue on this interval.
\end{thm}
The behaviour $\kappa^4$ of the leading order of $E(\kappa)$ (for $\kappa$ sufficiently small) equals the weak coupling asymptotic of the ground state energy of one-dimensional Schrödinger operators with zero-mean potentials as was shown in \cite{simon1976bound}. Also, the binding requirement (that $\kappa$ should be sufficiently small) is similar to one of the two binding requirements that were found in \cite{martin2004stability} for the three-dimensional Coulomb system.

In Figure \ref{fig:figure_1b} the leading behavior of the discrete eigenvalue given in Theorem \ref{thm:main_1} is compared to a numerical calculation of the smallest discrete eigenvalue of $H_\kappa$. The numerical calculations are done using a similar method to what was presented in \cite{rosenthal1971solution}. The figure shows that they agree well for $\kappa$ below $0.25$.

\begin{figure}[!ht]
\centering
\subfloat[ ]{
\includegraphics{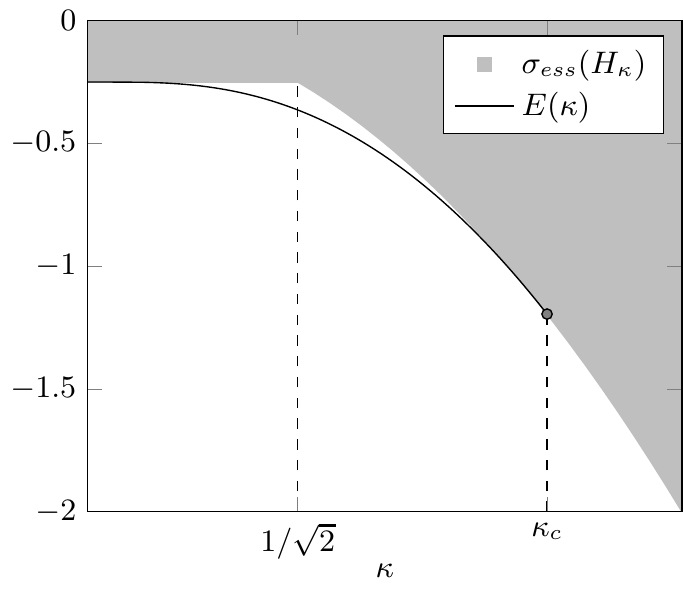}
\label{fig:figure_1a}}
\subfloat[ ]{
\includegraphics{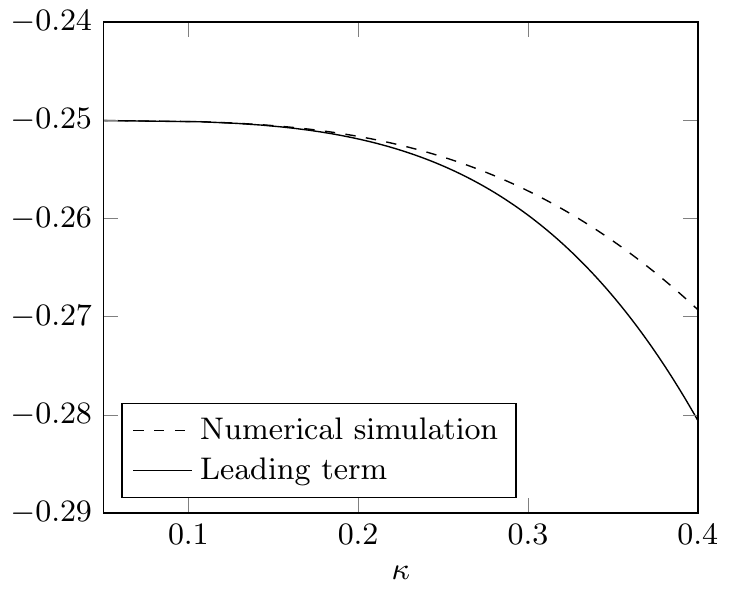}
\label{fig:figure_1b}}
\caption{In (a) a plot of the ground state energy is given as a function of the impurity charge $\kappa$. Figure (b) is a comparison of the leading term of the discrete eigenvalue, and the numerically calculated discrete eigenvalue.}
\end{figure}

The results can be generalized to hold for $0<\sigma<1$ as well. If $\kappa$ is sufficiently small the operator $H_{\kappa,\sigma}$ has a single discrete eigenvalue, and the leading behavior of this discrete eigenvalue $E$ is calculated to be
\begin{equation*}
	E(\kappa)= -\frac{1}{4(1-\sigma)} - \beta(\sigma)\kappa^4 +\mathcal{O}(\kappa^5),
\end{equation*}
where
\begin{equation}\label{eq:beta}
	\beta(\sigma) := 4\frac{\left[6 \sigma \sqrt{1 - \sigma^2} - (2 - \sigma) \sigma \pi - 8 \sigma^2 \cos^{-1}\left(\frac{\sqrt{1+\sigma}}{\sqrt{2}}\right) + \tan^{-1}\left(\frac{2\sigma(1-\sigma^2)}{1-2\sigma^2}\right)\right]^2}{(1+\sigma)(1-\sigma)^2\pi^2\sigma^2}
\end{equation}
when $0<\sigma<1/\sqrt{2}$. The solution can be extended to the range $1/\sqrt{2}\leq\sigma<1$ by choosing another branch of $\tan^{-1}$.

For $\kappa \geq 1/\sqrt{2}$ we have the following results.
\begin{thm}\label{thm:critical_kappa}
Let $H_{\kappa,\tilde\kappa}$ be the self-adjoint operator formally described by
\begin{equation}
H_{\kappa,\tilde\kappa} = -\frac{1}{2}\Delta -\delta(x-y) + \tilde\kappa\delta(x) - \kappa\delta(y).
\end{equation}
on $L^2(\R^2)$.
Given any $\tilde\kappa > 1$ there exists $\kappa_M$ such that $H_{\kappa,\tilde\kappa}$ has no discrete eigenvalues for all $\kappa \geq \kappa_M$. Furthermore, given any $0 <\tilde\kappa < 1$ there exists some $\kappa_M$ such that $H_{\kappa,\tilde\kappa}$ has at least one discrete eigenvalue for all $\kappa\geq \kappa_M$.
\end{thm}
As a consequence of the previous two theorems we will also prove the following corollary:
\begin{corollary}\label{col:kappa_c}
Let $H_\kappa$ be the operator in \eqref{eq:H_kappa}. Then there exists a critical charge of the impurity, which we will denote $\kappa_c$, such that the discrete spectrum of $H_{\kappa}$ is non-empty for all $0<\kappa < \kappa_c$ and empty for $\kappa \geq \kappa_c$.
\end{corollary}

Using numerical simulations to calculate the smallest discrete eigenvalue of $H_\kappa$ we see that at $\kappa \approx 1.546$ the ground state energy hits the essential spectrum. Thus, we expect that the true $\kappa_c$ is close to $1.546$. In Figure \ref{fig:figure_2} a numerical calculation of the critical charge $\kappa_c$ is plotted against the mass fraction $\sigma$. We see that as the mass of the impurity decreases the critical charge is increased, and thus bound states exists for impurities with larger charges. We have also plotted the coefficient $\beta$ in \eqref{eq:beta} against the mass fraction, and we see that the coefficient decreases as the mass of the impurity decreases.

\begin{figure}[!ht]
\centering
\includegraphics{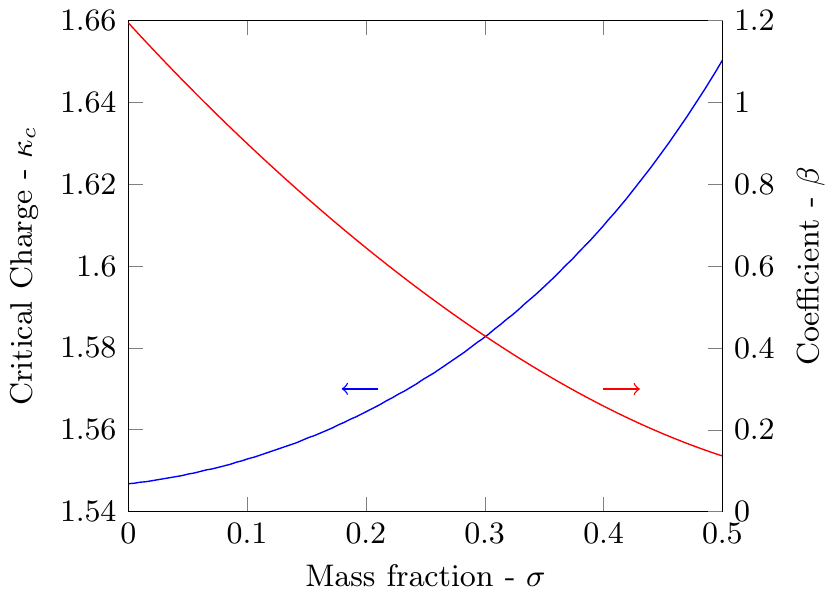}
\caption{Plot of the critical charge and the $\kappa^4$ coefficient of the discrete eigenvalue against the mass fraction $\sigma$.}
\label{fig:figure_2}
\end{figure}


\section{The Framework}\label{sec:framework}
In this section we introduce the framework we use to study the discrete spectrum of $H_{\kappa,\sigma}$. This framework has been used in \cite{cornean2008skeleton,cornean2006critical}, and we refer to those papers for more details.

We define $H_{\kappa,\sigma}$ as the unique self-adjoint operator associated to the sesqui-linear form
\begin{equation}
\begin{split}
Q(f,g) &= \frac{1}{2}\langle \mathbf{\nabla}f,A\nabla g\rangle_{L^2(\R^2)} - \langle f(x,x),g(x,x)\rangle_{L^2(\R)} \\
	& \quad + \kappa\langle f(0,y),g(0,y)\rangle_{L^2(\R)} - \kappa\langle f(x,0),g(x,0)\rangle_{L^2(\R)},
\end{split}
\end{equation}
on $H^1(\R^2) \times H^1(\R^2)$, where $H^1(\mathbb{R}^2)$ is the Sobolev space of first order and $A \in \R^{2\times 2}$ is the matrix
\begin{equation}
	A = \left[\begin{array}{c c} 1 & \sigma \\ \sigma & 1\end{array}\right].
\end{equation}
Let $\psi \in H^1(\mathbb{R}^2)$ and let $e \in \R^2$ be a unit vector. We define the trace operator $\tau_e: H^1(\mathbb{R}^2) \to L^2(\mathbb{R})$ by $(\tau_e\psi)(s) := \psi(s e)$. Let us write $\tau := (\tau_{e_1}, \tau_{e_2}, \tau_{e_{12}})$ as an operator defined on $H^1(\R^2)$ with values in $[L^2(\R)]^3 := \oplus_{i=1}^3 L^2(\R)$, where $\{e_1,e_2\}$ is the canonical basis in $\R^2$ and $e_{12} = 1/\sqrt{2}(e_1+e_2)$. Then $H_{\kappa,\sigma}$ is
\begin{equation}
	H_{\kappa,\sigma} = -\frac{1}{2}\Delta -\sigma\frac{\partial^2}{\partial x\partial y} + \tau^{*} g \tau,
\end{equation}
where $g := \diag\{-\kappa, \kappa, -1\} \in \R^{3 \times 3}$.

As a direct application of the Hunziker - Van Winter - Zhislin (HVZ) theorem\cite{simon1971quad} and a consequence of the signs of the potential terms in \eqref{eq:H_kappa} the following lemma holds.
\begin{lem}\label{lem:ess_spec}
The essential spectrum of $H_{\kappa,\sigma}$ is $$\left[\min\left\{-\frac{1}{4(1-\sigma)},-\frac{\kappa^2}{2}\right\},\infty\right).$$
\end{lem}
The essential spectrum of $H_{\kappa,0}$  is illustrated by the shaded area in Figure \ref{fig:figure_1a}.
Write the operator in \eqref{eq:H_kappa} as $H_{\kappa,\sigma} = H_{0,\sigma} - V_\kappa$, where
\begin{equation*}
H_{0,\sigma} := -\frac{1}{2}\Delta - \sigma \frac{\partial^2}{\partial x\partial y}, \qquad V_{\kappa} := \delta(x-y) -\kappa\delta(x) + \kappa\delta(y).
\end{equation*}
If $R(z)$ denotes the full resolvent operator $(H_{\kappa,\sigma} - z)^{-1}$ and $R_0(z)$ denotes the resolvent $(H_{0,\sigma} - z)^{-1}$, then by Krein's formula
\begin{equation}\label{eq:kreins}
	R(z) = R_0(z) - R_0(z)\tau^{*}(g^{-1}+\tau R_0(z)\tau^{*})^{-1}\tau R_0(z), \quad z\in \rho(H_{0,\sigma}) \cap \rho(H_{\kappa,\sigma}).
\end{equation}

Define:
\begin{equation}\label{eq:operator_pencil}
	G_{\kappa,\sigma}(z) := g^{-1} + \tau R_0(z) \tau^{*}.
\end{equation}
It can be shown that $E < \inf \sigma_{ess}(H_{\kappa,\sigma})$ belongs to the discrete spectrum of $H_{\kappa}$ if and only if $G_{\kappa,\sigma}(E)$ is not invertible. Note that $G_{\kappa,\sigma}(z)$ is a $3 \times 3$ operator valued matrix which acts on $[L^2(\R)]^3$ and its entries are $z$ dependent. We will denote the elements of $\tau R_0(z) \tau^{*}$ by
\begin{equation}\label{eq:fourier_t}
\tau R_0(z) \tau^{*} = \left[
\begin{array}{c c c}
	T_{0,\sigma} & T_{1,\sigma} & T_{2,\sigma}^{*} \\
	T_{1,\sigma} & T_{0,\sigma} & T_{2,\sigma}^{*} \\
	T_{2,\sigma} & T_{2,\sigma} & T_{3,\sigma}
\end{array}
\right].
\end{equation}
The integral kernel of $R_0(z)$ is
\begin{equation}\label{eq:resol_kernel}
	R_0(\mathbf{x},\mathbf{y},z) = \frac{1}{2\pi^2}\int_{\R^2} \frac{e^{i\mathbf{k}\cdot(\mathbf{x}-\mathbf{y})}}{|\mathbf{k}|^2 + 2\sigma k_1k_2-2z}\di{k_1}\di{k_2}.
\end{equation}
Using the integral kernel of $R_0(z)$ in the Fourier representation we can explicitly calculate the integral kernels of the elements in $\tau R_0(z)\tau^{*}$ (the first and the the last operators are multiplication operators in Fourier space):
\begin{align}
\hat{T}_{0,\sigma}(s) &= \frac{1}{\sqrt{(1-\sigma^2)s^2-2z}} \label{eq:hat_T_0}\\
\hat{T}_{1,\sigma}(s,t) &= \frac{1}{\pi}\frac{1}{s^2+t^2+2\sigma st-2z}\label{eq:hat_T_1}\\
\hat{T}_{2,\sigma}(s,t) &= \frac{1}{\pi}\frac{1}{t^2+(s-t)^2+2\sigma t(s-t)-2z}\label{eq:hat_T_2}\\
\hat{T}_{3,\sigma}(s) &= \frac{1}{\sqrt{(1-\sigma^2)s^2-(1-\sigma)4z}}\label{eq:hat_T_3}.
\end{align}
From these expressions it is easy to see that the operators in \eqref{eq:fourier_t} are bounded if ${\rm Re}(z)<0$, and their norms go to zero when ${\rm Re}(z) \to -\infty$.


\section{Proof of Theorem \ref{thm:main_1}}\label{sec:small_kappa}
We are now ready to prove the first of our main results, i.e. the existence of a single discrete eigenvalue of $H_\kappa = H_{\kappa,0}$ when $\kappa$ becomes sufficiently small. In the following we will also denote $T_{i,0}$ by $T_i$.

Assume that $\kappa < 1/\sqrt{2}$. In that case it follows from Lemma \ref{lem:ess_spec} that any discrete eigenvalues $E \in \R$ must satisfy $E < - 1/4$. Moreover, $E$ is a discrete eigenvalue of $H_\kappa$ if and only if the operator $G_{\kappa}(E)$ is not invertible. Define $\tilde{G}_\kappa(E) := \kappa^{-1}G_\kappa(E)$ for $\kappa >0$, then $\tilde{G}_\kappa(E)$ is invertible when $G_\kappa(E)$ is invertible. In matrix representation we can write $\tilde{G}_\kappa(E)$ as
\begin{equation}
	\tilde{G}_\kappa(E) = \left[\begin{array}{c c c}
	-\mathbbm{1} & 0 & 0 \\
	0 & \mathbbm{1} & 0 \\
	0 & 0 & 0
	\end{array}\right] + \kappa \left[
	\begin{array}{c c c}
	T_{0} & T_{1} & T_{2}^{*} \\
	T_{1} & T_{0} & T_{2}^{*} \\
	T_{2} & T_{2} & -\mathbbm{1} + T_{3}
	\end{array}
	\right],
\end{equation}
where $\mathbbm{1}$ denotes the identity operator on $L^2(\R)$. In order to find the values $E$ where the inverse of $\tilde G_\kappa(E)$ does not exist, we use Feshbach's formula (see Equations (6.1) and (6.2) in \cite{nenciu1991dynamics}) to reduce the dimension of the operator pencil we are trying to invert.

 Let $\Pi$ be the orthogonal projection such that $\Pi [L^2(\R)]^3$ is isomorphic to $L^2(\R)$, and $\Pi \tilde{G}_\kappa(E) \Pi \cong -\kappa +\kappa T_{3}$. The congruence symbol simply means that $\Pi \tilde{G}_\kappa(E) \Pi$ can be identified with $-\kappa +\kappa T_{3}$ on $L^2(\R)$. Let $\Pi^\perp := \mathbbm{1}-\Pi$ correspond to the orthogonal subspace $\Pi^\perp[L^2(\R)]^3$ which is isomorphic to $[L^2(\R)]^2$. Then, we get
\begin{equation}
\Pi^\perp \tilde{G}_\kappa(E)\Pi^\perp \cong
\left[\begin{array}{c c}
	-\mathbbm{1} & 0 \\
	0 & \mathbbm{1} \\
	\end{array}\right] +
	\kappa\left[\begin{array}{c c}
	T_{0} & T_{1} \\
	T_{1} & T_{0} \\
	\end{array}\right].
\end{equation}
The next Lemma gives conditions under which the inverse of $\Pi^\perp \tilde{G}_\kappa(E)\Pi^\perp$ exists as an operator on $\Pi^\perp[L^2(\R)]^3$.

\begin{lem}\label{thm:existence-R}
There exists $ K > 0$ such that  $R(E):= [\Pi^\perp \tilde{G}_\kappa(E)\Pi^\perp]^{-1}$ exists in $\Pi^\perp[L^2(\R)]^3$ for all $E < -1/4$ and $0 < \kappa < K$.
\end{lem}
\begin{proof}
We rewrite $\Pi^\perp \tilde{G}_\kappa(E)\Pi^\perp$ as
\begin{equation}
	\Pi^\perp \tilde{G}_\kappa(E)\Pi^\perp \cong \left(\left[\begin{array}{c c}
	\mathbbm{1} & 0 \\
	0 & \mathbbm{1} \\
	\end{array}\right] +
	\kappa\left[\begin{array}{c c}
	T_{0} & T_{1} \\
	T_{1} & T_{0} \\
	\end{array}\right]\left[\begin{array}{c c}
	-\mathbbm{1} & 0 \\
	0 & \mathbbm{1} \\
	\end{array}\right]\right)\left[\begin{array}{c c}
	-\mathbbm{1} & 0 \\
	0 & \mathbbm{1} \\
	\end{array}\right].
\end{equation}
The operators $T_{0}$ and $T_{1}$ are uniformly bounded on $L^2(\R)$ for $E < -1/4$. Thus, we can choose a constant $K > 0$ such that
\begin{equation}
	\left\|\kappa\left[\begin{array}{c c}
	T_{0} & T_{1} \\
	T_{1} & T_{0} \\
	\end{array}\right]\right\| < 1,
\end{equation}
for all $E < -1/4$ and $0 < \kappa < K$. Then the inverse $\Pi^\perp \tilde{G}_\kappa(E)\Pi^\perp$ exists for all $0 < \kappa < K$ and $E < -1/4$. Additionally, we can write $R(E)$ as a Neumann series
\begin{equation}\label{eq:neumann_series}
	R(E) \cong \left[\begin{array}{c c}
		-\mathbbm{1} & 0 \\ 0 & \mathbbm{1}
	\end{array}\right] + \sum_{j=1}^\infty (-1)^j\kappa^j\left[\begin{array}{c c}
		-\mathbbm{1} & 0 \\ 0 & \mathbbm{1}
	\end{array}\right]\left(\left[\begin{array}{c c}
		T_{0} & T_{1} \\ T_{1} & T_{0}
	\end{array}\right]\left[\begin{array}{c c}
		-\mathbbm{1} & 0 \\ 0 & \mathbbm{1}
	\end{array}\right]\right)^j,
\end{equation}
for all $0 < \kappa < K$.
\end{proof}

By Feshbach's formula and Lemma \ref{thm:existence-R} there exists $K$ sufficiently small such that if $0<\kappa<K$ and  $E < -1/4$, the inverse of $\tilde{G}_\kappa(E)$  exists if and and only if the inverse of
\begin{equation}
	S_W(E) = \Pi \tilde{G}_\kappa(E) \Pi - \Pi \tilde{G}_\kappa(E) \Pi^\perp R(E)\Pi^\perp \tilde{G}_\kappa(E) \Pi,
\end{equation}
exists as an operator restricted to the proper subspace. Using the matrix representation we can write $S_W(E)$ as
\begin{equation}\label{eq:S_W1}
	S_W(E) \cong \mathbbm{1} - T_{3} - \kappa \left[\begin{array}{c c}
		T_{2} & T_{2} \end{array}\right]
	R(E)\left[\begin{array}{c}
		T_{2}^{*} \\ T_{2}^{*}
	\end{array}\right], \quad \mbox{on } L^2(\R).
\end{equation}
Note that the contribution to $S_W(E)$ from the first term of $R(E)$ in \eqref{eq:neumann_series} is zero. To find the values where the inverse of $S_W(E)$ does not exist on $L^2(\R)$ we use the following version of the Birman-Schwinger\cite{simon1979functional} principle.

\begin{prop}\label{thm:first-BirSchw}
Let $E<-1/4$ and let $S_W(E)$ be given by \eqref{eq:S_W1}. There exist two bounded operators $V_1:[L^2(\R)]^2\to L^2(\R)$ and $V_{2}:L^2(\R) \to [L^2(\R)]^2$ such that $S_W(E)^{-1}$ exists  if and only if the inverse of
\begin{equation}\label{eq:BS-operator}
	\mathbbm{1}_2 -\kappa V_{2}(\mathbbm{1}-T_{3})^{-1}V_1
\end{equation}
exists on $[L^2(\R)]^2$, where $\mathbbm{1}_2$ is the identity operator on $[L^2(\R)]^2$. We call the operator in \eqref{eq:BS-operator} for the Birman-Schwinger operator.
\end{prop}
\begin{proof}
Let $\Psi \in L^2(\R)$ and define $V_{2} : L^2(\R) \to [L^2(\R)]^2$ as
\begin{equation}\label{hc-100}
	V_{2}\Psi = R(E)\left[\begin{array}{c}
		T_{2}^{*}\Psi \\ T_{2}^{*}\Psi
	\end{array}\right].
\end{equation}
By the boundedness of $R(E)$ and $T_2^*$ it follows that $V_2$ is a bounded operator. Furthermore, let $\Psi = [\Psi_1,\Psi_2] \in [L^2(\R)]^2$ and define the operator $V_1:[L^2(\R)]^2 \to L^2(\R)$ by
\begin{equation}\label{hc-101}
	V_1\Psi = T_{2}\Psi_1 + T_{2}\Psi_2.
\end{equation}
The operator $V_1$ is bounded since $T_{2}$ is bounded. Using $V_1$ and $V_{2}$ it is possible to rewrite the operator $S_W(E)$ on $L^2(\R)$ as
\begin{align*}
	S_W(E) = \mathbbm{1} - T_{3} - \kappa V_1V_{2} = (\mathbbm{1} - \kappa V_1V_{2}(\mathbbm{1}-T_{3})^{-1})(\mathbbm{1}-T_{3})\label{eq:S_W_rewrite},
\end{align*}
since the bounded inverse of $\mathbbm{1}-T_{3}$ exists on $L^2(\R)$ for all $E<-1/4$. Consequently $S_W(E)$ exists if and only if $(\mathbbm{1}-\kappa V_1 V_{2}(1-T_{3})^{-1})^{-1}$ exists on $L^2(\R)$. But for any fixed $\kappa$ we can choose $E$ sufficiently negative such that $\|\kappa V_1V_{2}(\mathbbm{1}-T_{3})^{-1}\|<1$ and we can expand in a Neumann series
\begin{equation*}
 	(\mathbbm{1} - \kappa V_1V_{2}(\mathbbm{1}-T_{3})^{-1})^{-1} = \sum_{j=0}\kappa^j[V_1V_{2}(\mathbbm{1}-T_{3})^{-1}]^j.
\end{equation*}
Using resummation, we obtain that if $E$ is sufficiently negative we have
\begin{equation}\label{eq:S_W_rewrite_two}
	S_W(E)^{-1} = (\mathbbm{1} - T_{3})^{-1} + \kappa(\mathbbm{1}-T_{3})^{-1}V_1\left(\mathbbm{1}_2 - \kappa V_2(\mathbbm{1}-T_3)^{-1}V_1\right)^{-1}V_{2}(\mathbbm{1}-T_{3})^{-1}.
\end{equation}
Both the left-hand and the right-hand side define meromorphic functions for ${\rm Re}(E) < -1/4$, hence we can use the right-hand side to extend $S_W(E)^{-1}$ everywhere where the Birman-Schwinger operator exists. This proves one implication. 

Conversely, if we define
\begin{equation*}
  	A := \kappa V_2(\mathbbm{1} - T_3)^{-1}V_1,
\end{equation*}
equation \eqref{eq:S_W_rewrite_two} implies:
\begin{align*}
	\kappa V_2 S_W(E)^{-1} V_1 = A + A(\mathbbm{1}-A)^{-1}A =  -\mathbbm{1} + (\mathbbm{1}-A)^{-1} 
\end{align*}
or
\begin{equation}\label{eq:A}
 	(\mathbbm{1}-A)^{-1} = \mathbbm{1} + \kappa V_2 S_W(E)^{-1}V_1.
 \end{equation}
Now we can extend $(\mathbbm 1-A)^{-1}$ using the right-hand side.  This concludes the proof.
\end{proof}

Let $V_1$ and $V_{2}$ be as in the above proof.  Then the discrete eigenvalues $E$ of $H_\kappa$ for $0<\kappa<K$ are those $E<-1/4$ for which the inverse of the Birman-Schwinger operator \eqref{eq:BS-operator} does not exist on $[L^2(\R)]^2$. In Fourier representation the operator $(\mathbbm{1} - T_{3})^{-1}$ is given by multiplication with
\begin{equation}\label{eq:multiplication_operator}
	\frac{1}{\sqrt{2\pi}}\left(1 - \frac{1}{\sqrt{s^2 - 4E}}\right)^{-1} = \frac{1}{\sqrt{2\pi}}\frac{2}{s^2-4E-1} +\frac{1}{\sqrt{2\pi}} +\frac{1}{\sqrt{2\pi}}\frac{1}{\sqrt{s^2-4E}+1}.
\end{equation}

The first term on the right hand side has a singularity at $E = -1/4$. As $\kappa$ becomes small any possible discrete eigenvalues will be close to $-1/4$, and thus we expect the singular term to be the significant contribution. To simplify notation we define $\varepsilon := -4E - 1>0$. Taking the Fourier transform of each term on the right-hand side of \eqref{eq:multiplication_operator} we get the integral kernel of $(\mathbbm{1} - T_{3})^{-1}$:
\begin{align}
	(\mathbbm{1} - T_{3})^{-1}(x,y) &= \frac{1}{\sqrt{\varepsilon}}e^{-\sqrt{\varepsilon}|x-y|} + \delta(x-y) + \frac{1}{2\pi}\int_\R \frac{e^{is(x-y)}}{\sqrt{s^2+\varepsilon+1}+1}\di{s} \nonumber \\
	&= \frac{1}{\sqrt{\varepsilon}} - \int_0^{|x-y|}e^{-\sqrt{\varepsilon}s}\di{s}+ \delta(x-y) + \frac{1}{2\pi}\int_\R \frac{e^{is(x-y)}}{\sqrt{s^2+\varepsilon+1}+1}\di{s}. \label{eq:integral_1-T3}
\end{align}
From \eqref{eq:integral_1-T3} we see that there are four contributions to $V_2(\mathbbm{1}-T_3)^{-1}V_1$. We will show that the operators that we get from the three last terms in \eqref{eq:integral_1-T3} are uniformly bounded for $\varepsilon > 0$. Only the second term may pose problems due to its linear growth, while the third term is the distribution kernel of the identity operator and the fourth term is multiplication by a uniformly bounded function in Fourier space for $\varepsilon > 0$. 

We show that the operator corresponding to the second term is uniformly bounded. By the construction of $V_1$ and $V_2$ the contribution that might be problematic is the operator with the integral kernel
\begin{equation}
	0\leq C(x,y) = \int_{\R^2}T_2^*(x,t)\left (\int_0^{|t-t'|}e^{-\sqrt{\varepsilon}s}\di{s}\right ) T_2(t',y)\di{t}\di{t'},
\end{equation}
since the other factors from $V_1$ and $V_2$ are bounded. We will show that $C(x,y)$ is the integral kernel of a Hilbert-Schmidt operator. To do that, we need the following result which is based on the Paley-Wiener theorem\cite{reed1975ii}.
\begin{lem}\label{lem:bounded_b1}
There exists $\alpha>0$ sufficiently small such that the kernels $T_2(x,y)e^{\alpha|x|}$, $T_2^*(x,y)e^{\alpha|y|}$, $T_1(x,y)e^{\alpha|y|}$ and $T_1(x,y)e^{\alpha|x|}$ are in $L^2(\R^2)$ uniformly in $\varepsilon > 0$. 
\end{lem}
\begin{proof}
We will show that $T_2(x,y)e^{\alpha|x|} \in L^2(\R^2)$ using the Paley-Wiener theorem. The proofs for the other integral kernels are similar and therefore not included. To apply Paley-Wiener we must show that $\hat{T}_{2}(s,t)$ can be analytically continued to a subset of the type $$\{\xi \in \C^2: |\imag(\xi)| < a\} \subset \C^2,$$ for some $a > 0$. Write $s=s_1+is_2$ and $t=t_1+it_2$, with $s_1,s_2,t_1,t_2 \in \R$, then
 \begin{equation*}
 	\hat{T}_{2}(s,t) = \frac{1}{\pi}\frac{1}{t_1^2-t_2^2+2it_1t_2+(s_1-t_1)^2-(s_2-t_2)^2+2i(s_1-t_1)(s_2-t_2)+\varepsilon +1}.
 \end{equation*}
 This function has no poles for $t_2$ and $s_2$ satisfying $t_2^2+(s_2-t_2)^2 < 1$, and is analytic on the subset
 \begin{equation*}
 	\left\{\xi\in \C^2 : |\imag(\xi)| < \frac{1}{2}\right\}\subset \C^2.
 \end{equation*}
Take $\eta = (s_2,t_2) \in \R^2$ such that $|\eta|<1/2$ and define $\delta:=\varepsilon +1 -t_2^2-(s_2-t_2)^2$. By the choice of $\eta$ we get $\delta > 0$, and the norm
\begin{equation*}
 	\|\hat{T}_{2}(\cdot + i\eta)\|_{L^2(\R^2)}^2 \leq \frac{1}{\pi^2}\int_{\R^2}\frac{1}{\left(t_1^2+(s_1-t_1)^2 + \delta\right)^2}\di{s_1}\di{t_1} = \frac{1}{\pi\delta} < \infty.
 \end{equation*}
Thus, $\|\hat{T}_{2,0}(\cdot + i\eta)\|_{L^2(\R^2)} < \infty$ for all such $\eta \in \R^2$. Then the Paley-Wiener theorem implies that $e^{\alpha\sqrt{x^2+y^2}}T_2(x,y) \in L^2(\R^2)$ for all $\alpha<1/2$. This concludes the proof of $T_2(x,y)e^{\alpha|x|} \in L^2(\R^2)$.
\end{proof}
We are now ready to show that $C(x,y)$ is an integral kernel of a Hilbert-Schmidt operator. To do that we use the following inequality
\begin{equation}\label{eq:last-term}
	C(x,y) \leq \int_{\R^2}T_2^*(x,t)|t|T_2(t',y)\di{t}\di{t'} +\int_{\R^2}T_2^*(x,t)|t'|T_2(t',y)\di{t}\di{t'},
\end{equation}
which follows from the definition of $C(x,y)$ and the inequality
\begin{equation*}
	\int_0^{|t-t'|}e^{-\sqrt{\varepsilon}s}\di{s}  \leq |t-t'| \leq |t| + |t'|.
\end{equation*}
We will show that the last the term in \eqref{eq:last-term} is in $L^2(\R^2)$ (the proof that the first term is also in $L^2(\R^2)$ is identical). Note that the integral is separable and
\begin{equation*}
	\int_{\R^2}T_2^*(x,t)|t'|T_2(t',y)\di{t}\di{t'} =: F(x)G(y).
\end{equation*}
We will show that $F,G \in L^2(\R)$. Applying the Cauchy-Schwarz inequality with respect to the $t$-integral and using Lemma \ref{lem:bounded_b1} we find
\begin{align}\label{hc-103}
	\|F\|_{L^2(\R)}^2 &= \int_\R\left|\int_\R T_2^*(x,t)\di{t}\right|^2\di{x} = \int_\R\left|\int_\R T_2^*(x,t)e^{\alpha|t|}e^{-\alpha|t|}\di{t}\right|^2\di{x} \nonumber \\
	& \leq C_\alpha \int_{\R^2}|T_2^*(x,t)|^2e^{2\alpha|t|}\di{t}\di{x} < \infty,
\end{align}
for $\alpha > 0$ sufficiently small. Similarly:
\begin{align*}
	\|G\|^2_{L^2(\R)} &= \int_\R\left|\int_\R|t'| T_2(t',y)\di{t'}\right|^2\di{y} \leq C_\alpha \int_\R\left|\int_\R e^{\frac{\alpha}{2}|t'|} T_2(t',y)\di{t'}\right|^2\di{y} \\
	&=C_\alpha \int_\R\left|\int_\R e^{-\frac{\alpha}{2}|t'|}e^{\alpha|t'|} T_2(t',y)\di{t'}\right|^2\di{y}  \leq \tilde{C}_\alpha \int_{\R^2}e^{2\alpha|t'|}|T_2(t',y)|^2\di{t'}\di{y} < \infty,
\end{align*}
again for $\alpha > 0$ sufficiently small. We conclude that $C(x,y) \in L^2(\R^2)$ uniformly in $\varepsilon > 0$. 

Using the expansion in \eqref{eq:integral_1-T3} the integral kernel of the Birman-Schwinger operator is
\begin{equation}\label{horia1}
	\mathbbm{1}_2 - \kappa \int_{\R^2}V_{2}(x,t)[(\mathbbm{1}-T_{3})^{-1}](t,t')V_1(t',y)\di{t}\di{t'} = \mathbbm{1}_2 - \frac{\kappa}{\sqrt{\varepsilon}}|\Psi\rangle \langle \Phi| + \kappa B_\varepsilon(x,y),
\end{equation}
where $B_\varepsilon(x,y)$ is the integral kernel of the uniformly bounded operator for $\varepsilon > 0$ that comes from the non-singular terms of \eqref{eq:integral_1-T3}. Also:
\begin{align}\label{eq:almostthere}
	\overline{\Phi(y)} &:= \int_\R V_1(x,y)\di{x},\quad \Psi(x) := \int_\R V_{2}(x,y)\di{y}.
\end{align}
The functions $\Psi$ and $\Phi$ are in $L^2(\R)$ and let us prove this for $\Psi$.  From the above definition and from \eqref{hc-100} we see that it is enough to prove that $\int_{\mathbb{R}} T_2^*(x,t)dt$ belongs to $L^2(\mathbb{R})$. But this is exactly what we did in \eqref{hc-103}.

By the usual factorization trick, the operator in \eqref{horia1} is invertible if and only if 
$$\mathbbm{1}_2 - \frac{\kappa}{\sqrt{\varepsilon}}|\Psi\rangle \langle \Phi| (\mathbbm{1}_2 + \kappa B_\varepsilon)^{-1}$$
is invertible. The later operator is not invertible if and only if $\varepsilon$ is a zero of  the following function 
$$(0,\infty)\ni \varepsilon \mapsto 1-\frac{\kappa}{\sqrt{\varepsilon}}\langle \Phi|(\mathbbm{1}_2 + \kappa B_\varepsilon)^{-1}|\Psi\rangle\in \mathbb{R}.$$
Introduce the new variable $r^2=\varepsilon$. The above function has a positive root $\varepsilon_0$ if and only if the map 
$$[-1,1]\ni r\mapsto f_\kappa(r):=\kappa\langle \Phi|(\mathbbm{1}_2 + \kappa B_{r^2})^{-1}|\Psi\rangle\in [-1,1]$$ 
has a positive fixed point $r_0>0$ and $r_0=\sqrt{\varepsilon_0}$. 

It is not difficult to extend the methods we used for proving that $B_\varepsilon$ was uniformly bounded in $\varepsilon>0$ in order to show that actually all the $\varepsilon$ dependent quantities are norm differentiable with globally bounded derivatives on $\varepsilon>0$. Thus $f_\kappa$ becomes a contraction if $\kappa$ is small enough and its unique fixed point $r_0$ can be computed by iteration starting from $r=0$. 

Using the definitions of $V_1$ and $V_{2}$ (in which we put $\varepsilon=0$ or equivalently $E=-1/4$) we can calculate the inner product $\langle \Phi,\Psi\rangle$ to get
\begin{equation*}
	\langle \Phi,\Psi\rangle|_{\varepsilon=0} = 8\kappa\left(\frac{4}{\pi} - 1\right) + \mathcal{O}(\kappa^2)>0.
\end{equation*}
Thus $r_0\sim \kappa \langle \Phi,\Psi\rangle \sim \kappa^2>0$ if $\kappa$ is small enough which leads to 
$\varepsilon_0 =r_0^2\sim \kappa^2 \langle \Phi,\Psi \rangle^2\sim \kappa^4$. Consequently,  the leading order behaviour of the discrete eigenvalue $E(\kappa)$ of $H_{\kappa}$ for $\kappa$ sufficiently small is
\begin{equation*}
	E(\kappa)= -\frac{1}{4}-16\left(\frac{4}{\pi}-1\right)^2\kappa^4 +\mathcal{O}(\kappa^5).
\end{equation*}
where we used the formula $\varepsilon = -4E - 1>0$. This concludes the first part of the proof of Theorem \ref{thm:main_1}.

We will now prove that the ground state energy is always non-degenerate (when it exists) by first showing that the heat semigroup $e^{-t H_{\kappa}}$ is positivity improving.  Some key formulas from \cite{albeverio1995fundamental} give the explicit expression of the heat kernel of $-\di^2/\di y^2 +\kappa \delta(y)$ from which we conclude that the integral kernel of
$$e^{-t (-\frac{1}{2}\Delta +\kappa \delta(y))}(x,y;x',y')=e^{t\frac{1}{2}\frac{\di^2}{\di x^2}}(x,x')e^{-t(-\frac{1}{2}\frac{\di^2}{\di y^2} +\kappa \delta(y))}(y,y'),\quad t>0,$$
is positive and point-wise smaller than $e^{t\frac{1}{2}\Delta}(x,y;x',y')$.  Applying the analogue of the Dyson formula between $e^{-t H_{\kappa}}$ and $e^{-t (-\frac{1}{2}\Delta +\kappa \delta(y))}$ (one has to be careful when deriving it due to the singularity of the delta "potentials") we see that the integral kernel of
$e^{-t H_{\kappa}}$ is larger or equal than that of $e^{-t (-\Delta +\kappa \delta(y))}$, hence it is also positivity improving. The Perron-Frobenius theorem\cite{reed1978iv} then guarantees the non-degeneracy of the lowest eigenvalue of $H_{\kappa}$, provided that such an eigenvalue exists.

In order to prove that a discrete eigenvalue exists for all $\kappa \in (0,1/\sqrt{2})$ we first need to extend our previous analysis to negative $\kappa$'s. It is not difficult to see from the expression of $H_\kappa$ that the previous existence result also holds for small negative $\kappa\neq 0$ as well. The family $H_\kappa$ is analytic of type B in the sense of Kato. The regular analytic perturbation theory allows one to extend the construction of a real analytic ground state energy $E(\kappa)$ from a neighborhood of $\kappa\neq 0$ to some maximal open intervals $I_\pm$ respectively included in  $(0,1/\sqrt{2})$ and $(-1/\sqrt{2},0)$. The only reason for which the right endpoint of $I_+$ might not go all the way to $1/\sqrt{2}$ is that $E(\kappa)$ might start increasing and eventually hit the bottom of the essential spectrum (i.e. $-1/4$) at some $\kappa_+<1/\sqrt{2}$.  We will show that this is not possible.

Fix $\epsilon>0$ small enough for which we know that $E(\pm \epsilon)$ exist. Then we can construct two families of real analytic normalized eigenvectors $\Psi_\kappa$ on $I_\pm$, starting from some given eigenvectors at $\kappa=\pm \epsilon$.

The operator which implements the interchange of $x$ with $y$ is denoted by $U$ and acts as $(Uf)(x,y)=f(y,x)$. It is unitary and $U=U^{-1}$. Moreover, we have
\begin{align*}
UH_{\kappa}U^{-1}=H_{-\kappa},\quad H_{\kappa}U^{-1}\Psi_{-\kappa}=E(-\kappa)U^{-1}\Psi_{-\kappa}.
\end{align*}
This shows that $E(-\kappa)$ is also an eigenvalue for $H_{\kappa}$, hence $E(\kappa)\leq E(-\kappa)$. By a similar argument we also obtain that $E(-\kappa)\leq E(\kappa)$, hence $E(\kappa)=E(-\kappa)$ as long as they exist. Moreover, there must exist a unimodular complex number $e^{i\phi(k)}$ (the phase can be chosen to be smooth on $|\kappa|>\epsilon$) such that
\begin{align}\label{hc-1}
\Psi_{\kappa}(x,y)= e^{i\phi(k)}\Psi_{-\kappa}(y,x),\quad \kappa\in I_\pm.
\end{align}

All the quantities defined above are smooth if $\kappa\neq 0$, but the eigenvectors are not a-priori $\kappa$-differentiable in the $H^1(\mathbb{R}^2)$ norm, only in $L^2(\R^2)$. We can formally apply the Feynman-Hellmann formula to the quadratic form and get:
\begin{align}\label{hc-3}
E'(\kappa)=\int_\R (|\Psi_\kappa(x,0)|^2-|\Psi_\kappa(0,x)|^2)dx.
\end{align}
The rigorous proof of this identity is based on the following identity 
$$\frac{1}{1+ \alpha E(\kappa)}=\langle \Psi_\kappa,  (1+\alpha H_{\kappa})^{-1}\Psi_\kappa\rangle, \quad 0< \alpha\ll 1,$$
in which we now can 
differentiate with respect to $\kappa$ in the norm topology and after that take the limit $\alpha\downarrow 0$.

We will now show that there cannot exist a $\kappa\in I_+$ such that $E'(\kappa)>0$. Assume the contrary and consider such a $\kappa$. Define the vector $\Phi(x,y)=\Psi_{-\kappa}(y,x)$ and choose $\kappa'\in I_+$ with $\kappa'>\kappa$. $\Phi$ is a normalized vector which belongs to the form domain of $H_{\kappa'}$. First using the min-max principle and second \eqref{hc-1} we have:
$$E(\kappa')\leq \langle \Phi,H_{\kappa'} \Phi\rangle =E(\kappa)+(\kappa'-\kappa)\int_\R (|\Phi(t,0)|^2-|\Phi(0,t)|^2)dt.
$$
Taking the limit $\kappa'\downarrow \kappa$ in $(E(\kappa')-E(\kappa))/(\kappa'-\kappa)$ leads to:
\begin{align}\label{hc-2}
E'(\kappa)\leq \int_\R (|\Phi(t,0)|^2-|\Phi(0,t)|^2)dt.
\end{align}
Due to \eqref{hc-1} we have $|\Phi(t,0)|^2=|\Psi_\kappa(0,t)|^2$ and $|\Phi(0,t)|^2=|\Psi_\kappa(t,0)|^2$, hence \eqref{hc-3} implies:
\begin{align}\label{hc-4}
\int_\R (|\Phi(t,0)|^2-|\Phi(0,t)|^2)dt=-\int_\R (|\Psi_\kappa(t,0)|^2-|\Psi_\kappa(0,t)|^2)dt=-E'(\kappa).
\end{align}
Introducing this identity back into \eqref{hc-2} we obtain $E'(\kappa) \leq 0$. We conclude that $E'(\kappa)\leq 0$ for all $\kappa\in I_+$, hence $E(\kappa)\leq E(\epsilon)<-1/4$ for $\kappa\in I_+$ which insures the existence of a positive minimal distance between $E(\kappa)$ and the essential spectrum. Consequently, the right endpoint $\kappa_+$ of $I_+$ cannot be smaller than $1/\sqrt{2}$ because in that case $E(\kappa_+):=\lim_{\kappa\uparrow \kappa_+}E(\kappa)$ would be an eigenvalue, thus $I_+$ could be extended a bit to the right of $\kappa_+$ by analytic perturbation theory. Hence the operator $H_{\kappa}$ must have at least one eigenvalue for
$0<\kappa\leq  1/\sqrt{2}$. This concludes the proof of Theorem \ref{thm:main_1}.


\section{Proof of Theorem \ref{thm:critical_kappa}}\label{sec:large_kappa}
In this section we prove the second main result, namely that if $\tilde\kappa > 1$ is fixed, then  $H_{\kappa,\tilde\kappa}$ has no discrete eigenvalues for $\kappa$ in a connected neighborhood of $+\infty$. The proof is based on a similar method as used in Section \ref{sec:small_kappa}. 

Since $H_\kappa$ and $H_{\kappa,\tilde\kappa}$ only differ in the positive interaction term while the bottom of the essential spectrum is given by the negative interaction terms, we have that $\sigma_{ess}(H_{\kappa,\tilde\kappa}) = \sigma_{ess}(H_\kappa)$. We assume that $\kappa \geq 1/\sqrt{2}$. Then Lemma \ref{lem:ess_spec} implies:
\begin{equation*}
\sigma_{ess}(H_{\kappa,\tilde\kappa}) = \left[-\frac{\kappa^2}{2}, \infty\right).
\end{equation*}
The framework described in Section \ref{sec:framework} is easily generalized to the operator $H_{\kappa,\tilde\kappa}$. Consequently, $E < -\kappa^2/2$ is a discrete eigenvalue of $H_{\kappa,\tilde\kappa}$ if and only if the inverse of the operator $\mathcal{G}_{\kappa,\tilde\kappa}(E)$ does not exist on $[L^2(\R)]^3$, where $\mathcal{G}_{\kappa,\tilde\kappa}(E)$ is given by
\begin{equation*}
  \mathcal{G}_{\kappa,\tilde\kappa}(E) = g^{-1} + \tau R_0(E)\tau^*,
\end{equation*}
and $\tau R_0(E) \tau^{*}$ is as before but $g$ is changed to $ \diag\{-\kappa,\tilde\kappa,-1\}$. To study when the operator $\mathcal{G}_{\kappa,\tilde\kappa}(z)$ is invertible, we scale it using the unitary operator $U_\kappa$ which acts on $L^2(\R)$ by $[U_\kappa f](x) = \sqrt{\kappa}f(\kappa x)$. We have:
\begin{align*}
	[U_\kappa \hat{T}_{1}(E) U_\kappa^{*}f](x) &=  \frac{1}{\pi\kappa}\int_\R \frac{1}{x^2+y^2-\frac{2E}{\kappa^2}}f(y)\di{y}.
\end{align*}
Define a rescaled energy $\varepsilon := -2E/\kappa^2>1$. Thus $U_\kappa \hat{T}_{1}(E)U_\kappa^{*} = \frac{1}{\kappa}\hat{T}_{1}(-\varepsilon)$. Equivalent results hold for $\hat{T}_{0},\hat{T}_{2},\hat{T}_{2}^*$ and $\hat{T}_{3}$. Consequently, the operator $\mathcal{G}_{\kappa,\tilde\kappa}(E)$ is unitarily equivalent to the operator:
\begin{equation}\label{eq:G_kappa}
	G_{\kappa,\tilde\kappa}(-\varepsilon) := \left[
	\begin{array}{c c c}
	-\frac{1}{\kappa} & 0 & 0 \\ 0 & \frac{1}{\tilde\kappa} & 0 \\ 0 & 0 & -\mathbbm{1}
	\end{array}
	\right] +
	\frac{1}{\kappa}\left[
	\begin{array}{c c c}
	T_0(-\varepsilon) & T_{1}(-\varepsilon) & T_{2}^{*}(-\varepsilon) \\ T_{1}(-\varepsilon) & T_{0}(-\varepsilon) & T_{2}^{*}(-\varepsilon) \\ T_{2}(-\varepsilon) & T_{2}(-\varepsilon) & T_{3}(-\varepsilon)
	\end{array}
	\right].
\end{equation}

As mentioned the strategy we apply to show the absence of discrete eigenvalues is basically the same as in Section \ref{sec:small_kappa}, i.e. some applications of Feshbach's formula and the Birman-Schwinger principle. So we begin by choosing the orthogonal projection $\Pi$ on $[L^2(\R)]^3$ which satisfies
\begin{equation}
	\Pi G_{\kappa,\tilde\kappa}(-\varepsilon) \Pi \cong \frac{1}{\kappa}\left[\begin{array}{c c} -\mathbbm{1}+T_{0}(-\varepsilon) & T_{1}(-\varepsilon) \\ T_{1}(-\varepsilon) & \frac{\kappa}{\tilde\kappa}+T_0(-\varepsilon) \end{array}\right],
\end{equation}
on $\Pi[L^2(\R)]^3$. We will also need the projection on the orthogonal subspace of $\Pi[L^2(\R)]^3$, which is defined by $\Pi^\perp:=\mathbbm{1}-\Pi$.

\begin{lem}\label{lem:invertible_R}
Let $G_{\kappa,\tilde\kappa}(-\varepsilon)$ be given by \eqref{eq:G_kappa}. Then $R(\varepsilon):=[\Pi^\perp G_{\kappa,\tilde\kappa}(-\varepsilon) \Pi^\perp]^{-1}$
exists as a bounded operator on  the proper subspace for all $\varepsilon > 1$ and $\kappa > 1/\sqrt{2}$.
\end{lem}
\begin{proof}
By the definition of $\Pi^\perp$ we have $\Pi^\perp G_{\kappa,\tilde\kappa}(-\varepsilon) \Pi^\perp \cong -\mathbbm{1}+\kappa^{-1}T_3(-\varepsilon)$. We need to check the invertibility of   $\mathbbm{1}-\kappa^{-1}T_{3}(-\varepsilon)$ on $L^2(\R)$. In the Fourier representation, this operator is a multiplication operator with the function
\begin{equation}
	\left(1-\frac{1}{\kappa}\frac{1}{\sqrt{s^2+2\varepsilon}}\right)^{-1}.
\end{equation}
Thus, the norm $\|\kappa^{-1}T_{3}\| < 1/(\kappa\sqrt{2})\leq 1$ for all $\kappa \geq 1/\sqrt{2}$ and $\varepsilon >1$. Consequently, $\Pi^\perp G_\kappa(-\varepsilon) \Pi^\perp$ is invertible on $L^2(\R)$ for all $\kappa > 1/\sqrt{2}$ and $\varepsilon >1$.
\end{proof}

By Feshbach's formula and Lemma \ref{lem:invertible_R} the inverse of $G_{\kappa,\tilde\kappa}(-\varepsilon)$ exists if the inverse of 
\begin{equation}
  S_W(\varepsilon) := \Pi G_{\kappa,\tilde\kappa}(-\varepsilon)\Pi -  \Pi G_{\kappa,\tilde\kappa}(-\varepsilon)\Pi^\perp R(\varepsilon) \Pi^\perp G_{\kappa,\tilde\kappa}(-\varepsilon)\Pi
\end{equation}
exists as an operator on $[L^2(\R)]^2$. In order to simplify notation, we stop writing the explicit dependence on $\varepsilon$ of the various $T$-operators. We get the following expression for $S_W(\varepsilon)$:
\begin{equation}\label{eq:S_W}
	S_W(\varepsilon) \cong \left[\begin{array}{c c} -\mathbbm{1}+T_{0} & T_{1} \\ T_{1} & \frac{\kappa}{\tilde\kappa}+T_{0} \end{array}\right] + \frac{1}{\kappa}\left[\begin{array}{c} T_{2}^{*} \\ T_{2}^{*} \end{array}\right]\left(\mathbbm{1} - \frac{1}{\kappa}T_{3}\right)^{-1}\left[\begin{array}{c c}T_{2} & T_{2}\end{array}\right].
\end{equation}
To find the conditions for the inverse of $S_W(\varepsilon)$ to exist on $[L^2(\R)]^2$, we apply Feshbach's formula again. Consequently, we need to define another pair of orthogonal projections $\tilde\Pi$ and $\tilde\Pi^\perp:= \mathbbm{1}-\tilde\Pi$ on $[L^2(\R)]^2$ such that
\begin{equation}
	\tilde\Pi S_W(\varepsilon)\tilde\Pi \cong -\mathbbm{1} + T_{0} + \frac{1}{\kappa}T_{2}^*(\mathbbm{1} - \kappa^{-1}T_{3})^{-1}T_{2} \quad \mbox{on} \quad L^2(\R).
\end{equation}
\begin{lem}\label{lem:invertible_Rtilde}
Let $S_W(\varepsilon)$ be given by \eqref{eq:S_W}, and let $\tilde\Pi^\perp$ be the orthogonal projection on $[L^2(\R)]^2$ such that
\begin{equation}
	\tilde\Pi^\perp S_W(\varepsilon)\tilde\Pi^\perp \cong \frac{\kappa}{\tilde\kappa} + T_{0} + \frac{1}{\kappa}T_{2}^*(\mathbbm{1} - \kappa^{-1}T_{3})^{-1}T_{2},
\end{equation}
on $L^2(\R)$. Then $\tilde R(\varepsilon):= [\tilde\Pi^\perp S_W(\varepsilon)\tilde\Pi^\perp]^{-1}$ exists on the proper subspace for all $\kappa \geq 1/\sqrt{2}$ and $\varepsilon > 1$.
\end{lem}
\begin{proof}
The proof follows from the fact that $T_{0}$ and $\kappa^{-1} T_{2}^*(\mathbbm{1} - \kappa^{-1}T_{3})^{-1}T_{2}$ are bounded and positive for all $\kappa \geq 1/\sqrt{2}$ and $\varepsilon > 1$.
\end{proof}

Lemma \ref{lem:invertible_Rtilde} and Feshbach's formula implies that the inverse of $G_{\kappa,\tilde\kappa}(\varepsilon)$ exists if the inverse of
\begin{equation}
	\tilde{S}_W(\varepsilon) \cong \mathbbm{1}-T_{0} -\frac{1}{\kappa}T_{2}^{*}(\mathbbm{1}-\kappa^{-1}T_{3})^{-1}T_{2} + W_{\kappa,\tilde\kappa}(\varepsilon)
\end{equation}
exists as an operator on $L^2(\R)$, where
\begin{align}
	W_{\kappa,\tilde\kappa}(\varepsilon) &:= D\left(\frac{\kappa}{\tilde\kappa}+T_{0}+\frac{1}{\kappa}T_{2}^*\left(\mathbbm{1}-\frac{1}{\kappa}T_{3}\right)^{-1}T_{2}\right)^{-1}D, \label{hc-10}\\
 D &:= T_{1} +\frac{1}{\kappa}T_{2}^{*}\left(\mathbbm{1}-\frac{1}{\kappa}T_{3}\right)^{-1}T_{2}.\label{hc-11}
\end{align}
The idea is to apply the Birman-Schwinger principle to study for which values of $\varepsilon > 1$ and $\kappa \geq 1/\sqrt{2}$ the inverse of $\tilde S_W(\varepsilon)$ does not exist on $L^2(\R)$. Before we do that we rewrite $\tilde S_W(\varepsilon)$ a bit. Factorizing $\kappa/\tilde\kappa$ in $W_{\kappa,\tilde\kappa}(\varepsilon)$ we can write
\begin{equation}\label{eq:S_W-tilde}
  \tilde S_W(\varepsilon) \cong \mathbbm{1} - T_0 +\frac{1}{\kappa}\widetilde{W}_{\kappa,\tilde\kappa},
\end{equation}
where
\begin{equation}
  \widetilde{W}_{\kappa,\tilde\kappa} = -T_2^*\left(\mathbbm{1}-\frac{1}{\kappa}T_3\right)^{-1}T_2 + \tilde\kappa D\left(\mathbbm{1}+\frac{\tilde\kappa}{\kappa}T_{0}+\frac{\tilde\kappa}{\kappa^2}T_{2}^*\left(\mathbbm{1}-\frac{1}{\kappa}T_{3}\right)^{-1}T_{2}\right)^{-1} D.
\end{equation}

We are now ready to construct the Birman-Schwinger operator for $\tilde S_W(\varepsilon)$ given by \eqref{eq:S_W-tilde}.
\begin{prop}\label{thm:birman-schwinger}
Let $\tilde{S}_W(\varepsilon)$ be as in \eqref{eq:S_W-tilde}, and let $\tilde\kappa > 0$ be fixed, $\kappa \geq 1/\sqrt{2}$ and $\varepsilon > 1$. Then there exists bounded operators $V_1:[L^2(\R)]^2 \to L^2(\R)$ and $V_2:L^2(\R) \to [L^2(\R)]^2$ such that  $\tilde{S}_W(\varepsilon)$ is invertible if and only if
\begin{equation}\label{eq:birman-schwinger}
	 \mathbbm{1}_2+ \frac{1}{\kappa}V_2 (\mathbbm{1} - T_{0})^{-1} V_1
\end{equation}
is invertible on $[L^2(\R)]^2$.
\end{prop}
\begin{proof}
The proof is almost identical to the proof of Theorem \ref{thm:first-BirSchw}, so we will only describe the construction of $V_1:[L^2(\R)]^2 \to L^2(\R)$ and $V_2:L^2(\R) \to [L^2(\R)]^2$. We need $V_1$ and $V_2$ to have the property that
\begin{equation*}
	V_1V_2 = \tilde{W}_{\kappa,\tilde\kappa}.
\end{equation*}
Let $\Psi \in L^2(\R)$ and let $D$ be as in \eqref{hc-11}. Define the operator $V_2:L^2(\R) \to [L^2(\R)]^2$ by
\begin{equation}\label{hc-12}
	V_2\Psi = \left[\begin{array}{c}
		-\left(\mathbbm{1}-\frac{1}{\kappa}T_{3}\right)^{-\frac{1}{2}}T_{2}\Psi \\
	\left(\mathbbm{1}+\frac{\tilde\kappa}{\kappa}T_{0}+\frac{\tilde\kappa}{\kappa^2}T_{2}^*\left(\mathbbm{1}-\frac{1}{\kappa}T_{3}\right)^{-1}T_{2}\right)^{-1/2}D \Psi
	\end{array}\right].
\end{equation}
Similarly, let $\Phi = [\Phi_1, \Phi_2]^T \in [L^2(\R)]^2$. We define the operator $V_1:[L^2(\R)]^2 \to L^2(\R)$ by
\begin{align}\label{hc-13}
	V_1\Phi &= \left[\begin{array}{c c} T_{2}^*\left(\mathbbm{1}-\frac{1}{\kappa}T_{3}\right)^{-\frac{1}{2}}, & \tilde\kappa D\left(\mathbbm{1}+\frac{\tilde\kappa}{\kappa}T_{0}+\frac{\tilde\kappa}{\kappa^2}T_{2}^*\left(\mathbbm{1}-\frac{1}{\kappa}T_{3}\right)^{-1}T_{2}\right)^{-1/2}\end{array}\right]\left[\begin{array}{c} \Phi_1 \\ \Phi_2 \end{array}\right] \nonumber \\
	& = T_{2}^*\left(\mathbbm{1}-\frac{1}{\kappa}T_{3}\right)^{-\frac{1}{2}}\Phi_1 + \tilde\kappa D\left(\mathbbm{1}+\frac{\tilde\kappa}{\kappa}T_{0}+\frac{\tilde\kappa}{\kappa^2}T_{2}^*\left(\mathbbm{1}-\frac{1}{\kappa}T_{3}\right)^{-1}T_{2}\right)^{-1/2}\Phi_2.
\end{align}
For $\Psi \in L^2(\R)$ we find that $V_1V_2\Psi$ is given by
\begin{equation*}
	V_1V_2\Psi = \tilde{W}_{\kappa,\tilde\kappa}\Psi
\end{equation*}
and we have our factorization.
\end{proof}

The strategy to show an absence of discrete eigenvalues is to find a necessary condition which any eigenvalue must satisfy, and then show that for every fixed $\tilde\kappa >1$ and for any $\kappa$ larger than some value $\kappa_M$ (depending on $\tilde\kappa$) the above necessary condition cannot be satisfied.  

The first important remark is that both $V_1$ and $V_2$ have finite limits when $\kappa\to\infty$, uniformly in $\epsilon>1$. Thus the operator in \eqref{eq:birman-schwinger} is always invertible if $\epsilon$ is larger than some value $\varepsilon_\kappa>1$. Moreover, this $\varepsilon_\kappa$ converges to $1$ when $\kappa$ goes to infinity. Therefore we know a priori that the points where \eqref{eq:birman-schwinger} might not be invertible on $[L^2(\R)]^2$ must obey  $\varepsilon \in(1,2)$ if $\kappa$ is larger than some value $\kappa_1$. Let us expand the integral kernel of $(\mathbbm{1}-T_0)^{-1}$ around the threshold $\varepsilon = 1$ and introduce the variable $\lambda$ (see below) to find the following
\begin{equation}\label{eq:T0kernel}
  (\mathbbm{1} - T_{0})^{-1}(x,y) = \frac{1}{\lambda} - |x-y| + \delta(x-y) + \frac{1}{2\pi}\int_\R \frac{e^{is(x-y)}}{\sqrt{s^2+1}+1}\di{s} + \mathcal{O}(\lambda),\quad \lambda:= \sqrt{\varepsilon - 1}.
\end{equation}
Using this expansion of the integral kernel, The Birman-Schwinger operator \eqref{eq:birman-schwinger} can be written as
\begin{equation}\label{eq:birman-schwinger-operator}
	\mathbbm{1}_2 + \frac{1}{\kappa}\frac{| \Psi\rangle  \langle \Phi |}{\lambda} + \frac{1}{\kappa}B,
\end{equation}
where the operator $B$ is given by the product of $V_2$, the non-singular terms of \eqref{eq:T0kernel} and $V_1$. Using the same approach as in Sec. \ref{sec:small_kappa}, we can show that $B$ is uniformly bounded for $\lambda > 0$ and $\kappa \geq 1/\sqrt{2}$. Furthermore, $|\Psi \rangle$ and $\langle \Phi |$ in \eqref{eq:birman-schwinger-operator} is given by
\begin{equation}
  |\Psi \rangle := \int_\R V_2(x,x')\di{x'}, \quad \langle \Phi | := \int_\R V_2(y',y)\di{y'},
\end{equation}
and $\Psi$ and $\Phi$ can be shown to be in $L^2(\R^2)$ using Lemma \ref{lem:bounded_b1}. 
Let us rewrite the Birman-Schwinger operator in \eqref{eq:birman-schwinger-operator}:
\begin{equation}
  \mathbbm{1}_2 + \frac{1}{\kappa}\frac{| \Psi\rangle  \langle \Phi |}{\lambda} + \frac{1}{\kappa}B = \left(\mathbbm{1}_2 + \frac{1}{\kappa}\frac{| \Psi\rangle  \langle \Phi |}{\lambda}\left[\mathbbm{1}_2 + \frac{1}{\kappa}B\right]^{-1}\right)\left(\mathbbm{1}_2 + \frac{1}{\kappa}B\right).
\end{equation}
But since $B$ is uniformly bounded in both $\lambda>0$ and $\kappa>1/\sqrt{2}$, there exists some $\kappa_2\geq \kappa_1>1/\sqrt{2}$ such that if $\kappa>\kappa_2$ we have that $\left(\mathbbm{1}_2 + \kappa^{-1}B\right)^{-1}$ exists on $[L^2(\R)]^2$ for all $\lambda > 0$. Consequently, for $\kappa >\kappa_2$ the inverse of the Birman-Schwinger operators exists at $\lambda\in (0,1)$ if and only if
\begin{equation}\label{eq:inverse_op}
  \left(\mathbbm{1}_2 + \frac{1}{\kappa}\frac{| \Psi\rangle  \langle \Phi |}{\lambda}\left[\mathbbm{1}_2 + \frac{1}{\kappa}B\right]^{-1}\right)^{-1},\quad \kappa>\kappa_2
\end{equation}
exists. Using Feshbach's formula with a rank-$1$ projection constructed from  $|\Psi\rangle$ we get  the only values of $0<\lambda<1$ where \eqref{eq:inverse_op} might not exist are those   which solve
\begin{equation}\label{eq:eig_sol}
  \lambda + \frac{1}{\kappa}\left\langle \Phi, \left[\mathbbm{1}_2 + \frac{1}{\kappa}B\right]^{-1}\Psi \right\rangle = 0,\quad \kappa>\kappa_2.
\end{equation}
Thus if $\kappa>\kappa_2$, any discrete eigenvalue of $H_{\kappa,\tilde\kappa}$ has to have a corresponding $\lambda \in (0,1)$ which is a solution to \eqref{eq:eig_sol}.

Let us define the function:
$$f(\lambda,\kappa):=\left\langle \Phi, \left[\mathbbm{1}_2 + \frac{1}{\kappa}B\right]^{-1}\Psi \right\rangle ,\quad \lambda\in [0,1],\quad \kappa\geq \kappa_2.$$
We are interested in finding possible values of $\lambda\in (0,1)$ where the graphs of $f(\lambda,\kappa)$ and $-\kappa \lambda$ cross each other. The function $f$ is jointly uniformly continuous. Moreover, by explicit computation we obtain:
\begin{align}\label{hc-20}
  \lim_{\kappa\to\infty}f(0,\kappa)= 2\pi\left(\tilde\kappa\int_\R \hat{T}_1(0,s)^2\di{s}-\int_\R \hat{T}_2^*(0,s)\hat{T}_2(s,0)\di{s} \right)\Big \vert_{\lambda=0} = \tilde\kappa -1. 
\end{align}
Thus there exists $\kappa_3>\kappa_2$ such that 
\begin{align*}
  f(0,\kappa)\geq  (\tilde\kappa -1)/2,\quad \kappa>\kappa_3.  
\end{align*}
From the uniform continuity of $f$ we obtain the existence of some $\delta\in (0,1)$ such that
\begin{align}\label{hc-21}
  f(\lambda,\kappa)\geq  (\tilde\kappa -1)/4>0,\quad \lambda\in (0,\delta),\quad \kappa>\kappa_3.  
\end{align}
  Moreover, $|f(\lambda,\kappa)|$ is bounded by some constant $K$ for all $\lambda$ and $\kappa$. This implies that if $\kappa>\kappa_3$, the value of $f(\cdot,\kappa)$ is positive on $(0,\delta)$ and is larger than $-K$ on $[\delta,1)$. At the same time, $-\lambda \kappa$ is negative on $(0,\delta)$ and less than $-\kappa \delta$ on $[\delta,1)$. Define $\kappa_M=\max\{\kappa_3, K/\delta\}$. Then the two graphs cannot intersect each other if $\kappa>\kappa_M$ and this completes the proof of absence of eigenvalues.

Now let us consider the case $0 < \tilde\kappa < 1$. All our previous considerations remain true up to and including the identity \eqref{hc-20} where now $\tilde\kappa -1<0$, hence 
\begin{align}\label{hc-22}
  f(0,\kappa)\leq (\tilde\kappa -1)/2<0,\quad \kappa>\kappa_3.  
\end{align}
Also, as before, $f(\lambda,\kappa)\geq -K$ for all $\lambda$ and $\kappa$. 

Consider the function $g(\lambda,\kappa)=\lambda\kappa +f(\lambda,\kappa)$ with $\lambda\in [0,1]$. We have $g(0,\kappa)=f(0,\kappa)<0$ while $g(1,\kappa)=\kappa +f(1,\kappa)\geq \kappa -K>0$ provided $\kappa>K$. Thus $g(\cdot,\kappa)$ must have a zero in $(0,1)$, and this proves the existence of discrete spectrum for all $\kappa>K$.

\section{Proof of Corollary \ref{col:kappa_c}} 
\label{sec:proof_of_corollary_col:kappa_c}
We can now prove the final result, i.e. the existence of a critical charge $\kappa_c$ which has the property that for every $0<\kappa<\kappa_c$ the operator $H_\kappa$ has at least one discrete eigenvalue, while if $\kappa\geq \kappa_c$ the discrete spectrum is empty. 

The proof has three steps. First, we show that  there exists some $\kappa_1\geq 1/\sqrt{2}$ such that $H_{\kappa_1}$ has no discrete spectrum.   Second, we show that given such a $\kappa_1$, the operator $H_\kappa$ has empty discrete spectrum for all  $\kappa\geq \kappa_1$. Third, we show that $\kappa_c$ is the smallest of all such $\kappa_1$.  

\vspace{0.2cm}

\noindent {\bf Step 1}. Let $\kappa>1/\sqrt{2}$ and consider the operator $H_{\kappa,2}$, i.e. with $\tilde\kappa=2>1$. Theorem \ref{thm:critical_kappa} implies the existence of a $\kappa_M > 1/\sqrt{2}$ such that $H_{\kappa,2}$ has no discrete eigenvalues if  $\kappa > \kappa_M$. 

We know that the operators $H_\kappa$ and $H_{\kappa,2}$ have the same essential spectrum. Additionally, we have that
\begin{equation}\label{hc-26}
  H_\kappa \geq H_{\kappa,2}\quad \mbox{if } \kappa \geq 2,
\end{equation}
where the inequality should be understood in the sense of quadratic forms. If $\kappa_1=\kappa_M+1$, the operator $H_{\kappa_1,2}$ has no discrete spectrum, hence \eqref{hc-26} and the min-max principle imply that the discrete spectrum of $H_{\kappa_1}$ is empty.

\vspace{0.2cm}

\noindent {\bf Step 2}. We will now show that the discrete spectrum of $H_\kappa$ with $\kappa\geq \kappa_1$ is also empty.  Define the unitary operator $U_\kappa:L^2(\R^2) \to L^2(\R^2)$ by $(U_\kappa \Psi)(x,y) = \kappa \Psi(\kappa x,\kappa y)$. Then by direct calculation
\begin{equation*}
	U^{-1}_\kappa H_\kappa U_\kappa =\kappa^2 \widetilde{H}_\kappa,\quad \tilde{H}_\kappa :=  -\frac{1}{2}\Delta -\delta(y) + \delta(x) - \frac{1}{\kappa}\delta(x-y).
\end{equation*}

Using the HVZ theorem we can prove that for $\kappa \geq 1/\sqrt{2}$ the essential spectrum of $\tilde H_\kappa$ is $[-1/2,\infty)$. Additionally, due to the sign of the $\kappa$-dependent term we have
\begin{equation*}
	\widetilde{H}_{\kappa} \geq \widetilde{H}_{\kappa_1}, \quad \mbox{if }\kappa \geq \kappa_1.
\end{equation*}
The operator $\widetilde{H}_{\kappa_1}=\kappa_1^{-2} U_{\kappa_1}^{-1}H_{\kappa_1}U_{\kappa_1}$ has no discrete spectrum. Since the bottom of the essential spectrum of $\widetilde{H}_\kappa$ is constant in $\kappa$ and equals $-1/2$, the min-max principle implies that $\widetilde{H}_{\kappa}$ has no discrete spectrum and the same holds true for $H_\kappa$.

\vspace{0.2cm}

\noindent {\bf Step 3}. The set $S$ consisting of all the $\kappa_1$'s considered in the previous two steps is bounded from below by $1/\sqrt{2}$ due to Theorem \ref{thm:main_1}. Let $\kappa'\geq 1/\sqrt{2}$ be the infimum of $S$.  Assume that $\kappa'$ does not belong to $S$. Then there would exist a ground state with energy $E(\kappa')<-\kappa'^2/2$. Using the analytic perturbation theory we could extend this ground state energy to a small interval centered at  $\kappa'$, thus $\kappa'$ would not belong to the closure of $S$, contradiction.

\vspace{0.2cm}

Thus $S=[\kappa',\infty)$ and $\kappa_c=\kappa'$. In fact, this proof provides us with an alternative characterisation of $\kappa_c$, i.e. $\kappa_c$ is the right endpoint of the open interval of $\kappa$'s for which a ground state exists.  


\section{Conclusions} 
\label{sec:conclusion}
In this paper we considered the discrete spectrum of the Schr\"odinger operator for a one-dimensional three-body system with Dirac delta potentials, which models an impurity interacting with an exciton. We have proven that for $\kappa$ close to zero there exists a single non-degenerate bound state which behaves like $\kappa^4$, and we have explicitly calculated the coefficient of the leading term. The ground state survives when $\kappa \in (0,1/\sqrt{2})$, but for some charge $\kappa_c > 1/\sqrt{2}$ the ground state energy hits the essential spectrum, and no bound states of the system exists for $\kappa \geq \kappa_c$. We cannot give an explicit value for $\kappa_c$, but numerical calculations indicate that $\kappa_c \approx 1.546$.

A future project is to study a related system of an impurity and two oppositely charged particles with multiplicative potentials in both one and two dimensions. While the results are expected to be somehow similar, the technical tools one needs to use are quite different. 

\section*{Acknowledgements}
J.H. and T.G.P. are supported by the QUSCOPE Center, which is funded by the Villum Foundation.
H.C. was partially supported by the Danish Council of Independent Research | Natural Sciences, Grant DFF-4181-00042. H.K. was partially supported by the MIUR-PRIN2010-11 grant for the project ``Calcolo delle variazioni'' .
\bibliographystyle{unsrt}
\bibliography{references}{}
\end{document}